\documentclass[11pt]{article}

\usepackage{graphicx} 
\usepackage{amsthm, amsmath, amsfonts,mathtools, fullpage, graphicx, wrapfig, tikz, caption, subcaption}
\usepackage{natbib}
\usepackage{amssymb}
\usepackage{verbatim}
\usepackage{relsize}
\usepackage{listings}
\usepackage{xcolor}
\usepackage{hyperref}
\usepackage{graphicx}
\usepackage{multirow}
\usepackage[linesnumbered,ruled,vlined]{algorithm2e}
\newcommand\independent{\protect\mathpalette{\protect\independenT}{\perp}}
\def\independenT#1#2{\mathrel{\rlap{$#1#2$}\mkern2mu{#1#2}}}
\newtheorem{theorem}{Theorem}

\newtheorem{remark}{Remark}

\newtheorem{assumption}{Assumption}
\usepackage{booktabs}
\usepackage{tabularx}

\title{Trust Me, I'm a Doctor? }

\author{
    Zach Shahn\\
    Department of Epidemiology and Biostatistics\\
    CUNY Graduate School of Public Health and Health Policy
    \and
    Mats Stensrud\\
   Institute of Mathematics\\
    École Polytechnique Fédérale de Lausanne (EPFL)
}

\date{}

\begin{document}

\maketitle

\begin{abstract}
Clinical trials usually target average treatment effects, but treatment decisions are made for individuals. This tension motivates a common criticism of evidence-based medicine: a treatment that is beneficial on average may be inappropriate for a particular patient, and skilled physicians may outperform rigid adherence to the strategy that performed best in a randomized trial. We consider how randomized and observational data from the same target population can be used to assess that possibility. Specifically, we study settings in which trial and/or observational data yield estimates of average outcomes under `treat all', `treat none', and usual care strategies in the population of interest. We derive sharp bounds on the proportion of encounters with physicians whose personal strategies outperform the better of `treat all' or `treat none' under the assumption that no physician's strategy is worse than always choosing the worse performing of `treat all' and `treat none'. These results clarify when clinical data support relying on physician discretion over the trial-average recommendation and when stronger justification is required.
\end{abstract}

\section{Introduction}
A common 
critique of evidence-based medicine is that clinical trials estimate \textit{average} treatment effects across heterogeneous
populations, and that average effects are poor guides to treatment decisions for individual patients.
A treatment that is
beneficial on average may be harmful for a particular subgroup, and vice versa.
Skilled clinicians can in principle leverage knowledge of individual patients to
do better than rigid adherence to trial-derived
recommendations. \citet{deaton2018understanding} express this criticism forcefully (italics their own): ``If \textit{your} physician tells \textit{you} that she endorses evidence-based medicine, and that the drug will work for \textit{you} because an RCT has shown that ‘it works’, it is time to find a physician who knows that \textit{you} and the \textit{average} are not the same." However, patients are hard pressed to decide whether to trust a doctor who claims to be able to outperform the `naive' strategy of assigning everyone the treatment that performed best in a trial. Moreover, doctors may themselves be unsure whether to trust their instincts over `the evidence'.

Imagine that a recent RCT has shown that a certain major surgical operation is on average harmful, but your doctor recommends operating nonetheless based on their assessment of your personal characteristics. For example, in randomized trials, transcatheter aortic-valve replacement performed better than surgical valve replacement on important short-term composite outcomes \citep{mack2019transcatheter}, but in practice many surgeons and heart teams still prefer surgical replacement for some patients. In this work, we consider how one might bring evidence from trials and observational data to decide whether to place more trust in the average effect or personalized physician judgment.

Suppose that outcomes under usual care outperform the outcomes in both arms of a randomized experiment conducted in the same population. A textbook example concerning lung cancer patients comes from \citet{hernan2024causal}[Fine Point 22.7], see also \citet{sarvet2023perspective,sarvet2024rejoinder}. Then Deaton and Cartwright's argument that one should trust their physician over a trial is valid in that following the doctors' advice gives better outcomes on average than universally recommending the treatment that performed best in the trial. In this setting, some of the patients might still discern that \textit{their} particular doctor is bad or malicious, or that patients similar to them tend to receive poor care. A reviewer pointed out that this possibility should be taken seriously, especially given documented racist malpractice in the United States. 


Another possibility is that usual care is outperformed by both treatment and control arms. Then, patients identified by their doctors as particularly likely to be harmed by treatment would, on average, actually be more likely to benefit (and/or the converse). We rule out this option, which is implausible in many but not all settings, and thus focus on the remaining interesting scenario in which, without loss of generality, usual care outperforms control but is itself outperformed by treatment. In this scenario, we will derive a bound on the proportion of encounters with doctors having personal strategies that outperform `treat all' under the assumption that no doctor's personal strategy is worse than treating nobody. This bound can offer some guidance about whether to trust one's doctor.

\section{Setting}

Let $\mathcal P$ denote the target population of patient/doctor encounters for which we observe
usual care treatment decisions and outcomes. All expectations and
probabilities below are with respect to this target population unless otherwise
noted. We consider a binary treatment $A\in\{0,1\}$, where $A=1$ denotes
treatment and $A=0$ denotes control. The outcome $Y\in[0,1]$ is interpreted
as a bounded utility, with larger values preferred. $Y$ could represent a scaled clinical
outcome, such as disease free survival over a fixed horizon divided by the
length of that horizon, or a weighted average of (non-)occurrence indicators of multiple outcomes including side effects. A binary outcome is a
special case in which $Y=1$ denotes the desirable outcome and $Y=0$ denotes
the undesirable outcome.

We write $Y(1)$ and $Y(0)$ for the potential utilities under treatment and
control. We assume that the expected utilities under the static strategies
`treat all' and `treat none' are available for the target population of encounters
$\mathcal P$. That is, we assume that we can identify or consistently estimate
\[
V_T \equiv E_{\mathcal P}\{Y(1)\},
\qquad
V_C \equiv E_{\mathcal P}\{Y(0)\}.
\]
This would hold directly given a pragmatic randomized trial in a random sample from the
target population. More generally, it holds after standardizing or
transporting trial-arm mean outcomes to $\mathcal P$, provided the required
transportability assumptions hold, see, for example \citet{pearl2022external,dahabreh2019generalizing,graham2024towards}. If $S$ indicates trial
participation and $X$ denotes measured covariates sufficient for
transportability, one sufficient condition is
\[
E\{Y(a)\mid X,S=\text{trial}\}
=
E\{Y(a)\mid X\},
\qquad a\in\{0,1\},
\]
together with consistency and positivity of trial participation over the
covariate support of $\mathcal P$. Under such conditions,
\[
V_a
=
E_{\mathcal P}\left[
E\{Y\mid A=a,S=\text{trial},X\}
\right],
\qquad a\in\{0,1\}.
\]
Alternatively, given observational data under usual care alone, the relevant quantities are identified under the standard causal assumptions: positivity, consistency, and exchangeability \citep{hernan2024causal}. 

Usual care arises from the treatment strategies of a population of doctors.
Let $\mathcal D$ denote the doctor superpopulation. We write $D$ for a
doctor from the doctor population $\mathcal D$, drawn from a random encounter in $\mathcal P$. For a fixed doctor $d$, let $A^d$ denote the treatment that
doctor $d$ would assign to a randomly drawn patient from $\mathcal P$.
Define
\[
p_d \equiv P_{\mathcal P}(A^d=1),
\qquad
V^d \equiv E_{\mathcal P}\{Y(A^d)\}.
\]
In words, $p_d$ is the proportion of encounters that would be treated if all encounters followed the treatment strategy of doctor $d$. $V^d$ is the average outcome that would be observed if all encounters followed the treatment strategy of doctor $d$. $p_D$ and $V^D$ denote the random values of these quantities obtained by drawing a doctor randomly from the distribution of encounters. 

We assume exchangeable assignment of patients to doctors in the sense that doctors' case mixes are independent of their treatment rules. 
\begin{assumption}[No systematic association between doctor rules and case mix]\label{assumption:case_mix_exchangeability}
Let \(Z\) denote the patient characteristics relevant to treatment decisions
or potential utilities, so that \[
E[Y(a)|Z,D]=E[Y(a)|Z]
\] for all $a$. For an individual doctor $d$, let \(F_d\) denote the distribution of \(Z\) among
patients seen by $d$, and let $g_d: Z\mapsto \{0,1\}$ denote doctor $d$'s
treatment rule. Let $F_D$ and $g_D$ denote random values of these variables obtained by drawing doctors (and their strategies) from the population of encounters. We assume
\[
F_D \independent g_D .
\]
\end{assumption} 
\noindent This assumption could be violated, for example, if more patients who would benefit from treatment went to academic medical centers, and doctors at academic medical centers had more aggressive treatment rules due to more recent training.

Under Assumption \ref{assumption:case_mix_exchangeability}, the target population mean utility under usual care is
\[
V_U \equiv E[V^D] = E\{Y(A^D)\},
\]
and the target population proportion treated under usual care is
\[
p \equiv E[p_D] = P(A^D=1).
\]
Throughout the main analysis we focus on the case where
\[
V_T > V_U > V_C,
\]
so that the static strategy `treat all' performs best in the trial, usual care improves on
treating none, but usual care does not outperform treating everyone. We will also make the charitable assumption that no doctor's judgment based strategy performs worse than assigning control to everybody, i.e. 
\begin{assumption}\label{ass:no_bad_doctors}
$V^D\geq V_C$ almost surely.
\end{assumption} 
We focus on a setting in which the data do not contain doctor identifiers for each encounter and/or there is one encounter per doctor. Thereby we avoid the complications of incorporating information from repeated encounters, which we discuss briefly in Remark \ref{rem:ids}. 

\section{Some preliminary results}
In the case of a binary $Y$, \citet{tian2000probabilities} derived a bound on the proportion counterfactually harmed by treatment, $\pi_H\equiv Pr(Y(1)<Y(0))$:
\begin{equation}\label{eq:piH_bound}
\pi_H \in [\phi_L,\phi_U],
\end{equation}
where
\begin{align*}
\phi_L &= \max\{0,\,V_C - V_U\}, \\
\phi_U &= \min\left\{
\begin{aligned}
&V_C, \\
&1 - V_T, \\
&P(Y=1,A=0) + P(Y=0,A=1), \\
&V_C - P(Y=1,A=0) + \bigl(1 - V_T - P(Y=0,A=1)\bigr)
\end{aligned}
\right\}.\footnote{Other approaches, e.g. \citet{shahn2025identification}, tighten this bound to the point of identification under additional strong assumptions involving covariates.}
\end{align*}
 Instead of committing to a cross-world counterfactual notion of harm \citep{sarvet2023perspective,sawant2026counterfactual,mueller2025meaning,sarvet2024rejoinder}, we can also interpret the bounds on $\pi_H$ as bounds on how much better any individual doctor's strategy is compared to `treat everybody'. An ideal doctor would give treatment for all $i$ such that $Y_i(1)> Y_i(0)$ and control to all $i$ such that $Y_i(1)<Y_i(0)$. In that sense, when $Y$ is binary, we regard bounding $\pi_H$ as a first step toward understanding potential gains from doctor judgment. 

 We can also check the data for empirical evidence that doctors are fruitfully exploiting effect modifiers in making their decisions. Let $p \equiv P(A=1)$ denote the proportion treated in the usual care arm. Define
\begin{equation}\label{eq:gain}
G \equiv V_U - \bigl(pV_T + (1-p)V_C\bigr),
\end{equation}
which also can be expressed as $cov(A,Y(1)-Y(0))$. It is possible that doctors follow a range of treatment strategies based on a range of factors. Suppose, for example, that some doctors always treat, some never treat, and others treat only patients with red hair. If the doctors based decisions on variables that are completely non-informative, in the sense that these variables are not (qualitative) effect modifiers, then: the expected outcome among those treated in the usual care arm would be $V_T$; the expected outcome among those untreated in the usual care arm would be $V_C$; and $V_U$ would be equal to $pV_T + (1-p)V_C$. Also, if $V_U>pV_T + (1-p)V_C$, this implies that doctors' treatment decisions are correlated (in the right direction) with some effect modifier \citep{stensrud2024optimal}.  That is, doctors are, knowingly or not, incorporating some useful information into their treatment decisions. Thus, $G>0$ suggests that we might study doctors' reasoning about treatment decisions to learn how to design dynamic treatment strategies. On the other hand, $G<0$ does not necessarily imply that no doctors are incorporating effect modifiers into their decisions.

\section{When to trust the expert?}
Now, $G>0$ implies some doctors are making use of some effect modifiers, but it does not imply that any doctor is actually outperforming `treat all' overall, as would be desired to trust a doctor's recommendation over the best treatment in the trial. We now consider bounds on the proportion of encounters with such high performing doctors. 

For fixed $d$, let $\delta_d\equiv V^d-V_T$ denote the effect of strategy $d$ compared to `treat all', so that $\delta_D$ is a random variable when $D$ is drawn from $\mathcal{D}$. We now bound the proportion of encounters with doctors whose strategies outperform `treat all' by any given amount $\delta^*$, i.e.\ $P(\delta_D\geq\delta^*)$, where the probability is over the draw of $D$ from the population of encounters.
\begin{theorem}\label{theorem1}
 Suppose $Y\in[0,1]$, $V_T>V_U>V_C$, and Assumptions \ref{assumption:case_mix_exchangeability} and \ref{ass:no_bad_doctors} hold. Define
\[
h_{\max}\equiv \min\{V_C,\,1-V_T\}.
\]
For $0<\delta^*\leq h_{\max}$, 
\begin{equation}\label{eq:bound}
P_D(\delta_D\geq\delta^*)
\leq
\min\left\{
\frac{V_U-V_C}{V_T-V_C+\delta^*},
\frac{1-(V_T-V_C)+(V_U-V_C)-p}{2\delta^*}
\right\}.
\end{equation}
For $\delta^*=0$,
\begin{equation}\label{eq:bound_delta0}
P_D(\delta_D\geq\delta^*=0) = P_D(V^D\geq V_T)
\leq
\frac{V_U-V_C}{V_T-V_C}.
\end{equation}
For $\delta^*>h_{\max}$,
\[
P_D(\delta_D\geq\delta^*)=0.
\]
If $\delta^*>\phi_U$ and $Y$ is binary, then
\[
P_D(\delta_D\geq\delta^*)= P_D(V^D\geq V_T+\delta^*)=0.\footnote{When $\delta^*>\pi_H$ for binary $Y$, recall that $P(\delta_D>\delta^*)=0$ as $\pi_H$ is the bound on the performance of any individual doctor. However, we cannot observe $\pi_H$, and the data are consistent with values of $\pi_H8$ as high as the upper bound $\phi_U$.}
\]

The bounds above are sharp in the sense that they are obtained by a distribution compatible with the observed marginal expectations ($V_U$, $V_C$, $V_T$, and $p$). 
\end{theorem}
\begin{proof}
See Appendix A.
\end{proof}

The term $V_U-V_C$ appearing in (\ref{eq:bound}) and (\ref{eq:bound_delta0}) is sometimes referred to as a population intervention effect \citep{hubbard2008population,laurendeau2025improved}, which is a scaled version of the average treatment effect in the treated. Inequality (\ref{eq:bound_delta0}) says that the maximum proportion of encounters with doctors who do at least as well as `treat all' is the ratio of the population intervention effect to the ATE. 

\begin{remark}[Estimation and inference]
\label{rem:estimation}
The bound in Theorem \ref{theorem1} motivates some simple estimators. For any fixed $\delta^*>0$, we can replace $V_T$, $V_C$, $V_U$, and $p$ by consistent estimators in the expressions in \eqref{eq:bound} and then take the minimum. This is a consistent plug-in estimator under mild conditions.

Consider now the two expressions in the minimum in \eqref{eq:bound}, which we will denote $B_1(\delta^*)$ and $B_2(\delta^*)$. Suppose that
$$
\left|B_1(\delta^*)-B_2(\delta^*)\right|\geq \epsilon
$$
for some $\epsilon>0$. Under standard regularity conditions, we can use the conventional delta-method and bootstrap inference for inference. However, when $B_1(\delta^*)$ and $B_2(\delta^*)$ are equal, the minimum is nondifferentiable and these procedures are not necessarily valid.

When $\delta^*=0$, the bound is simply
$$
\frac{V_U-V_C}{V_T-V_C}.
$$
which is a smooth function of the underlying means when $V_T>V_C$, and we can use standard delta-method or bootstrap inference.
\end{remark}

\begin{remark}[Proportions of encounters vs. proportions of doctors]\label{rem:encounter_vs_doctor}
    All probability statements are given with respect to the population of encounters as opposed to the population of doctors. That is, we are bounding `proportion of encounters with doctors whose strategies are better than treat all' as opposed to `proportion of doctors whose strategies are better than treat all'.  If encounter volume varies across doctors, the encounter frame is arguably the relevant frame for a patient entering usual care. Suppose that better doctors tend to have more encounters, but patients are exchangeably sorted to doctors. Then the proportion of doctors with superior strategies would understate the probability that a given patient's doctor has a superior strategy, because patients are more likely to see doctors with more encounters, who in this scenario tend to be better. A doctor wondering whether to trust their personal judgment over a trial, however, would be interested in the proportion of doctors with superior strategies. The bounds we provide are therefore less directly relevant to doctors than patients. However, if all doctors have approximately the same number of encounters, the two frames approximately coincide and the bounds are relevant to all parties.
\end{remark}


\begin{remark}[Multiple identified observations per doctor]
\label{rem:ids}
The main analysis defines $D$ as the doctor attached to a randomly
drawn encounter and therefore concerns a distribution
of doctor strategies that is weighted by encounters. When doctor identifiers are recorded, we can also
consider the doctor frame. Let $\widetilde D$ denote a doctor drawn with
equal probability from the target population of doctors, or according to
appropriate doctor level sampling weights.

Suppose each sampled doctor contributes $n$ conditionally independent
encounters. To connect these observations to $V^d$, rather than merely to
doctor $d$'s performance on their own case mix, suppose either that each
doctor's encounters are drawn from $\mathcal P$ or that sufficient
information is available to standardize each doctor's observations to
$\mathcal P$. As $n\to\infty$, $V^d$ and $p_d$ can then be identified for
each sampled doctor. Consequently, the empirical distribution of
$(V^d,p_d)$ across sampled doctors is identified. If the number of
representatively sampled doctors is also large, this identifies the
distribution of $(V^{\widetilde D},p_{\widetilde D})$ in the doctor
population and hence
\[
 \Pr\!\left(V^{\widetilde D}\geq V_T+\delta^*\right)
\]
directly. The encounter weighted distribution considered in the main
analysis can likewise be recovered by weighting doctors according to
their encounter volumes.

For finite $n$, the doctor level distribution is generally only partially
identified because observed variation among doctor-specific averages
combines true between doctor heterogeneity with within doctor sampling
variation. Nevertheless, repeated encounters provide restrictions beyond
the marginal means used in Theorem~\ref{theorem1}. For example, let
$R_{dj}$ denote a case mix weighted encounter outcome satisfying
$E(R_{dj}\mid d)=V^d$. Conditional independence of encounters gives, for
$k\leq n$,
\[
 E\left\{\prod_{j=1}^k R_{\widetilde D j}\right\}
 =
 E\left\{(V^{\widetilde D})^k\right\}.
\]
Thus, $n$ encounters per doctor identify the first $n$ moments of the doctor weighted distribution of $V^d$. Sharp bounds on $\Pr(V^{\widetilde D}\geq V_T+\delta^*)$ could therefore be obtained by
optimizing over distributions on $[V_C,1]$ having these identified
moments. These bounds weakly tighten as $n$ increases and converge to the
target probability as the full moment sequence becomes identified.
Cross products involving treatment indicators can similarly identify
moments involving $p_{\widetilde D}$ and further restrict the joint
distribution of $(V^{\widetilde D},p_{\widetilde D})$. We do not derive
these intermediate bounds here.
\end{remark}

\begin{remark}[Refinements using the joint law of $(A,Y)$]
Theorem 1 gives a closed form bound based on $V_T$, $V_C$, $V_U$, and
$p$. This bound does not use all information in the usual care joint law
$P(A,Y)$. A sharper bound could be obtained by optimizing over compatible
joint distributions of $(Y(1),Y(0),A^D,D)$. Specifically, for binary $Y$ we can maximize
\[
P_D(V^D\ge V_T+\delta^*)
\]
subject to
\[
E[Y(1)]=V_T,\qquad E[Y(0)]=V_C,
\]
\[
P(A^D=a,Y(A^D)=y)=P(A=a,Y=y),\qquad a,y\in\{0,1\},
\]
and
\[
V^D\ge V_C \quad\text{almost surely}.
\]
This is an optimization over linear programs indexed by probability of harm after representing doctor rules by their
treatment probabilities within the four principal strata of
$(Y(1),Y(0))$. Its value is weakly smaller than the bound in Theorem \ref{theorem1},
because it imposes all constraints used in Theorem \ref{theorem1} plus additional
constraints from the joint law of $(A,Y)$. For continuous $Y$, additional complications arise. In either case, Theorem \ref{theorem1} should be viewed as a simple analytic bound rather than the
most informative bound obtainable.
\end{remark}

Examining (\ref{eq:bound}) and (\ref{eq:bound_delta0}), a larger ATE, a larger desired doctor
advantage $\delta^*$, and a smaller population intervention effect
all lead to smaller upper bounds on the proportion of encounters with
doctors beating `treat all'. When $\delta^*>0$, the usual-care
treatment probability $p$ may further tighten the bound. For example,
suppose that $V_T-V_C=0.25$ and $V_U-V_C=0.05$. Equation (\ref{eq:bound_delta0}) implies
that at most
\[
\frac{0.05}{0.25}=20\%
\]
of encounters are with doctors whose personal strategies are at least
as good as `treat all', regardless of $p$. For doctors outperforming
`treat all' by $0.05$, however, equation (\ref{eq:bound}) gives
\[
\begin{aligned}
\min\left\{
  \frac{0.05}{0.25+0.05},
  \frac{1-0.25+0.05-p}{2(0.05)}
\right\}
&=\min\{0.167,\,8-10p\}.
\end{aligned}
\]
Thus, the upper bound is approximately $17\%$ when $p=0.75$ because
the first component of the minimum binds, but only $10\%$ when $p=0.79$ because the
second component of the minimum binds.

On the other hand, suppose that $V_T-V_C=0.05$, $V_U=0.20$,
$V_C=0.17$, and $p=0.97$. Equation (\ref{eq:bound_delta0}) allows as many as
\[
\frac{0.03}{0.05}=60\%
\]
of encounters to be with doctors doing at least as well as `treat
all.' For doctors outperforming `treat all' by $0.02$, equation (\ref{eq:bound})
instead gives
\[
\begin{aligned}
\min\left\{
  \frac{0.03}{0.05+0.02},
  \frac{1-0.05+0.03-0.97}{2(0.02)}
\right\}
&=\min\{0.429,\,0.25\} \\
&=0.25.
\end{aligned}
\]
Thus, even though a relatively large proportion of encounters may be with doctors whose strategies are at least as good as `treat all' when the ATE is small, the observed treatment proportion can substantially restrict the proportion of encounters with doctors achieving a meaningful
improvement over `treat all'.

How exactly should these bounds inform a patient's decision whether to trust their doctor? We now briefly shift from frequentist to subjective probability. The bounds (\ref{eq:bound}) and (\ref{eq:bound_delta0}) apply to the fraction of encounters in the population who outperform `treat all', but a patient's degree of belief about their particular encounter is a subjective probability that may differ. To illustrate a point, we describe a simplified setting where there are two types of doctors. Proportion $\pi_{\delta^*}$ are `good' doctors with advantage $\phi_U>V^d-V_T=\delta^*>0$ and $1-\pi_{\delta^*}$  are `bad' doctors with $V^d-V_T=\frac{V_U-V_T-\pi_{\delta^*}\delta^*}{1-\pi_{\delta^*}}\equiv V_{bad}<0$ (the value implied by $V_U-V_T$ and $\delta^*$). We also suppose that all doctors have similar numbers of encounters. Then we can consider the bound (\ref{eq:bound}) as a generous prior probability that a random doctor is `good'. But suppose that a patient thinks highly of their doctor because their doctor has performed well in the past or went to a prestigious medical school. This patient does not accept $(\ref{eq:bound})$ as a bound on the probability that \textit{their} doctor is good. For a fixed $\delta^*$, let $\tilde{\pi}_{\delta^*}$ denote a patient's degree of belief that their doctor is good. To rationally trust one's doctor's recommendation over `treat all', it must be the case that $\tilde{\pi}_{\delta^*}>\frac{-V_{bad}}{\delta^*-V_{bad}}$. If $\frac{-V_{bad}}{\delta^*-V_{bad}}$ is much greater than bound (\ref{eq:bound}), this would mean that the patient should require very strong evidence of their doctor's superiority to update their prior sufficiently to trust. Also, the higher $\delta^*$ is (i.e. the better `good' doctors are), the lower the threshold probability $\frac{-V_{bad}}{\delta^*-V_{bad}}$ becomes because there is more to be gained from trusting a good doctor. However, higher $\delta^*$ also reduces the bound (\ref{eq:bound}), making it a priori less likely that a doctor should be trusted. 

\section{Discussion}
 To make our key points, we have ignored many complications that would occur in practice. We limited our analysis to treatments occurring at a single time point, excluding settings where `monitor the situation and then decide' is an option. Clinical skill might in reality often take the form of knowing \textit{when} to treat as opposed to whether to treat. We also assumed that doctors' case mixes were not associated with their treatment strategies under usual care, which is unlikely to be exactly true in reality unless patients are randomly or haphazardly assigned to doctors. We further ignored practical difficulties in defining continuous outcome utilities. For example, in the aortic-valve replacement example from the introduction, surgical replacement is often preferred in younger patients because they would derive more utility from resulting improved durability, that is, longer time to replacement, compared to older patients with less time to live \citep{mack2019transcatheter}. It is non-trivial to weigh such expected durability gains against short term realized outcomes to arrive at a single scalar utility. Furthermore, we did not consider doctor self assessment of their level of confidence, e.g. how to account for a doctor stating they are \textit{especially} sure that a particular patient would be harmed by treatment based on that patient's characteristics. Doctor confidence level is potentially important information that could be available to both doctors and patients at decision time. Relatedly, suppose that many doctors switch their strategies to `treat all' after the trial results are released. Then, encounters in which doctors withhold treatment would be higher confidence compared to such encounters under pre-trial usual care. This could limit the applicability of results based on usual care data collected before the trial.
 
Qualitatively, our results agree with common sense. A higher ATE in a trial and worse outcomes under usual care \textit{should} both make one less likely to accept a doctor’s recommendation over trial evidence. We agree with \citet{deaton2018understanding} that in many cases it is important to `find a physician who knows that you and the average are not the same.' However, it is also important to find a physician who is realistic about their ability to distinguish you from the average. We hope that by formalizing and quantifying some relevant considerations in a simplified setting, we encourage further work into the important question of how evidence based on averages should impact decisions for individuals.

\section*{Acknowledgments} We wish to thank Amit Sawant, Marco Piccininni, and Inna Kleyman for helpful discussions. We used GPT-5.5 and GPT-5.6 to revise the language and to assist with technical arguments and proofs. We have reviewed and edited all text after use of the AI tool and take full responsibility for the content of the manuscript.

\bibliographystyle{plainnat}
\bibliography{bibliography}

@article{sarvet2023perspective,
  title={Perspective on ‘harm’ in personalized medicine},
  author={Sarvet, Aaron L and Stensrud, Mats J},
  journal={American Journal of Epidemiology},
  volume={194},
  number={6},
  pages={1743--1748},
  year={2025},
  publisher={Oxford University Press}
}

@article{graham2024towards,
  title={Towards a unified theory for semiparametric data fusion with individual-level data},
  author={Graham, Ellen and Carone, Marco and Rotnitzky, Andrea},
  journal={arXiv preprint arXiv:2409.09973},
  year={2024}
}

@incollection{pearl2022external,
  title={External validity: From do-calculus to transportability across populations},
  author={Pearl, Judea and Bareinboim, Elias},
  booktitle={Probabilistic and causal inference: The works of Judea Pearl},
  pages={451--482},
  year={2022}
}

@article{dahabreh2019generalizing,
  title={Generalizing causal inferences from individuals in randomized trials to all trial-eligible individuals},
  author={Dahabreh, Issa J and Robertson, Sarah E and Tchetgen, Eric J and Stuart, Elizabeth A and Hern{\'a}n, Miguel A},
  journal={Biometrics},
  volume={75},
  number={2},
  pages={685--694},
  year={2019},
  publisher={Oxford University Press}
}

@article{sarvet2024rejoinder,
  title={Rejoinder to ``Perspectives on `harm' in personalized medicine--an alternative perspective''},
  author={Sarvet, Aaron L and Stensrud, Mats J},
  journal={American Journal of Epidemiology},
  volume={194},
  number={6},
  pages={1752--1755},
  year={2025},
  publisher={Oxford University Press}
}

@article{tian2000probabilities,
  title={Probabilities of causation: Bounds and identification},
  author={Tian, Jin and Pearl, Judea},
  journal={Annals of Mathematics and Artificial Intelligence},
  volume={28},
  number={1},
  pages={287--313},
  year={2000},
  publisher={Springer}
}

@book{hernan2024causal,
  title={Causal Inference: What If},
  author={Hernan, M.A. and Robins, J.M.},
  isbn={9781420076165},
  lccn={2022050839},
  series={Chapman \& Hall/CRC Monographs on Statistics \& Applied Probab},
  year={2024},
  publisher={CRC Press}
}

@article{shahn2025identification,
  title={Identification and Estimation of Joint Potential Outcome Distributions from a Single Study},
  author={Shahn, Zach and Madigan, David},
  journal={arXiv preprint arXiv:2509.20506},
  year={2025}
}

@article{deaton2018understanding,
  title={Understanding and misunderstanding randomized controlled trials},
  author={Deaton, Angus and Cartwright, Nancy},
  journal={Social science \& medicine},
  volume={210},
  pages={2--21},
  year={2018},
  publisher={Elsevier}
}

@article{stensrud2024optimal,
  title={Optimal regimes for algorithm-assisted human decision-making},
  author={Stensrud, Mats J and Laurendeau, Julien David and Sarvet, Aaron Leor},
  journal={Biometrika},
  volume={111},
  number={4},
  pages={1089--1108},
  year={2024},
  publisher={Oxford University Press}
}

@article{hubbard2008population,
  title={Population intervention models in causal inference},
  author={Hubbard, Alan E and Van der Laan, Mark J},
  journal={Biometrika},
  volume={95},
  number={1},
  pages={35--47},
  year={2008},
  publisher={Oxford University Press}
}

@article{laurendeau2025improved,
  title={Improved bounds and inference on optimal regimes},
  author={Laurendeau, Julien D and Sarvet, Aaron L and Stensrud, Mats J},
  journal={Journal of the American Statistical Association},
  pages={1--13},
  year={2025},
  publisher={Taylor \& Francis}
}

@article{mueller2025meaning,
  title={The meaning of “harm” in personalized medicine—an alternative perspective},
  author={Mueller, Scott and Pearl, Judea},
  journal={American Journal of Epidemiology},
  volume={194},
  number={6},
  pages={1749--1751},
  year={2025},
  publisher={Oxford University Press}
}

@article{sawant2026counterfactual,
  title={Counterfactual Harm: A Counter-argument},
  author={Sawant, Amit N and Stensrud, Mats J},
  journal={American Journal of Epidemiology},
  pages={kwag064},
  year={2026},
  publisher={Oxford University Press}
}

@article{mack2019transcatheter,
  title   = {Transcatheter Aortic-Valve Replacement with a Balloon-Expandable Valve in Low-Risk Patients},
  author  = {Mack, Michael J. and Leon, Martin B. and Thourani, Vinod H. and Makkar, Raj R. and Kodali, Susheel K. and Russo, Michael and Kapadia, Samir R. and Malaisrie, S. Chris and Cohen, David J. and Pibarot, Philippe and Leipsic, Jonathon and Hahn, Rebecca T. and Blanke, Philipp and Williams, Mathew R. and McCabe, James M. and Brown, David L. and Babaliaros, Vasilis and Goldman, Steven and Szeto, Wilson Y. and Genereux, Philippe and Pershad, Ashok and Pocock, Stuart J. and Alu, Maria C. and Webb, John G. and Smith, Craig R.},
  journal = {New England Journal of Medicine},
  volume  = {380},
  number  = {18},
  pages   = {1695--1705},
  year    = {2019},
  doi     = {10.1056/NEJMoa1814052}
}

\section*{Appendix}
\subsection*{Appendix A: Proof of Theorem \ref{theorem1}}

For any physician $d$, Assumption \ref{ass:no_bad_doctors}
states that $V^d\geq V_C$. Therefore, for $\delta^*\geq 0$,
\[
V_U
=
E_D[V^D]
\geq
P_D(\delta_D\geq\delta^*)(V_T+\delta^*)
+
\{1-P_D(\delta_D\geq\delta^*)\}V_C.
\]
Rearranging gives
\[
P_D(\delta_D\geq\delta^*)
\leq
\frac{V_U-V_C}{V_T+\delta^*-V_C}.
\]
This proves the first inequality in \eqref{eq:bound}
and the inequality in \eqref{eq:bound_delta0}.

It remains to establish the second inequality and sharpness. Let
\[
\Delta=V_T-V_C,\qquad
r=V_U-V_C,\qquad
X^d=V^d-V_C.
\]
For a physician with $X^d=x$, bounded outcomes imply
\[
|\Delta-x|
=
\left|
E\left[(1-A^d)\{Y(1)-Y(0)\}\right]
\right|
\leq 1-p_d.
\]
Consequently,
\[
p_d\leq p_d^{\max}(x)=
\begin{cases}
1-\Delta+x, & 0\leq x\leq\Delta,\\
1+\Delta-x, & x\geq\Delta.
\end{cases}
\]

Now suppose a proportion $\pi>0$ of physicians are
`good', with $X^D\geq\Delta+\delta^*$, where $\delta^*>0$.
Since $r<\Delta$, we have $\pi<1$.
Let $\bar X^g\geq\Delta+\delta^*$ denote their mean gain.
Their average treatment probability is at most
$1+\Delta-\bar X^g$.
Let $\bar X^b$ denote the mean gain among the remaining
$1-\pi$ physicians. To satisfy $E_D[X^D]=r$,
we must have
\[
\bar X^b
=
\frac{r-\pi\bar X^g}{1-\pi}.
\]
Because $p_d\leq 1-\Delta+X^d$ for every physician,
the average treatment probability among the remaining
physicians is at most $1-\Delta+\bar X^b$.

Thus, for a given $\pi$,
\[
\begin{aligned}
p
&\leq
\pi(1+\Delta-\bar X^g)
+
(1-\pi)
\left(
1-\Delta+\frac{r-\pi\bar X^g}{1-\pi}
\right)\\
&=
1-\Delta+r-2\pi(\bar X^g-\Delta).
\end{aligned}
\]
This upper bound is maximized by setting
$\bar X^g=\Delta+\delta^*$, giving
\[
p\leq 1-\Delta+r-2\delta^*\pi.
\]
Rearranging gives
\[
\pi
\leq
\frac{1-\Delta+r-p}{2\delta^*}.
\]
When $\pi=0$, this inequality also holds because
averaging $p_d\leq 1-\Delta+X^d$ gives
$p\leq 1-\Delta+r$.
This establishes the second inequality in \eqref{eq:bound}.

To show sharpness, suppose
$0<\delta^*\leq h_{\max}=\min\{V_C,1-V_T\}$ and fix
\[
\pi=
\min\left\{
\frac{r}{\Delta+\delta^*},
\frac{1-\Delta+r-p}{2\delta^*}
\right\}.
\]
Assign mass $\pi$ to physicians with gain
$X^g=\Delta+\delta^*$.
Assign the remaining mass $1-\pi$ to physicians
all having gain
\[
\bar X^b
=
\frac{r-\pi(\Delta+\delta^*)}{1-\pi}.
\]
Because $\pi\leq r/(\Delta+\delta^*)$, we have
$\bar X^b\geq 0$, and since $r<\Delta$,
we also have $\bar X^b<\Delta$.

To realize these gains, take binary potential outcomes with
\[
\begin{aligned}
P\{Y(1)=0,Y(0)=1\}&=\delta^*,\\
P\{Y(1)=1,Y(0)=0\}&=\Delta+\delta^*,\\
P\{Y(1)=1,Y(0)=1\}&=V_C-\delta^*,\\
P\{Y(1)=0,Y(0)=0\}&=1-V_T-\delta^*.
\end{aligned}
\]
Assign patients independently
of physicians, ensuring
Assumption \ref{assumption:case_mix_exchangeability}.

Good physicians treat all patients in the benefit
stratum and no patients in the harm stratum.
Varying treatment among unaffected patients gives
every treatment probability in
\[
p_g\in[\Delta+\delta^*,\,1-\delta^*],
\]
while preserving gain $\Delta+\delta^*$.

For a bad physician, treating only mass $\bar X^b$
from the benefit stratum gives gain $\bar X^b$
and treatment probability $\bar X^b$.
Alternatively, treating all harmed and unaffected
patients, together with mass $\bar X^b+\delta^*$
from the benefit stratum, gives the same gain
and treatment probability $1-\Delta+\bar X^b$.
Both rules are feasible because $\bar X^b<\Delta$.
Interpolating their stratum-specific treatment
probabilities yields every value in
\[
p_b\in[\bar X^b,\,1-\Delta+\bar X^b],
\]
while preserving gain $\bar X^b$.

For this construction, the minimum attainable average
treatment probability is
\[
\pi(\Delta+\delta^*)+(1-\pi)\bar X^b=r,
\]
and the maximum attainable average treatment probability is
\[
\pi(1-\delta^*)
+
(1-\pi)(1-\Delta+\bar X^b)
=
1-\Delta+r-2\delta^*\pi.
\]
By the definition of $\pi$,
\[
p\leq 1-\Delta+r-2\delta^*\pi.
\]
Moreover, any feasible observed data law satisfies
$p\geq r$, since
\[
X^d
=
E\left[A^d\{Y(1)-Y(0)\}\right]
\leq p_d.
\]
Therefore,
\[
p\in
\left[
r,\,
1-\Delta+r-2\delta^*\pi
\right].
\]
It follows that the treatment probabilities of the
good and bad physicians can be chosen within their
attainable intervals so that $E_D[p_D]=p$,
while preserving $E_D[X^D]=r$ and
\[
P_D(X^D\geq\Delta+\delta^*)=\pi.
\]
Thus the bound in \eqref{eq:bound} is attained.

For $\delta^*=0$, use the same
construction with $\delta^*=0$, take $\pi=r/\Delta$,
and assign gains $\Delta$ and $0$ to the good and bad
physicians, respectively.
The attainable average treatment probabilities then
range from $r$ to $1-\Delta+r$, which contains every
feasible $p$. The same argument therefore proves
sharpness of \eqref{eq:bound_delta0}.

Finally, for every physician,
\[
\begin{aligned}
\delta_d
&=
E\left[(1-A^d)\{Y(0)-Y(1)\}\right]\\
&\leq
E\left[\{Y(0)-Y(1)\}_+\right]\\
&\leq
\min\{V_C,\,1-V_T\}
=
h_{\max}.
\end{aligned}
\]
Hence $P_D(\delta_D\geq\delta^*)=0$ whenever
$\delta^*>h_{\max}$.
For binary $Y$,
\[
E\left[\{Y(0)-Y(1)\}_+\right]
=
\pi_H\leq\phi_U,
\]
which also establishes the zero-probability claim
when $\delta^*>\phi_U$.
\end{document}